\documentclass[aps,prx,twocolumn,showpacs]{revtex4-1}
\usepackage{bm}
\usepackage{mathrsfs}
\usepackage{amsmath}
\usepackage{amssymb}
\usepackage{graphicx}
\usepackage{amsfonts}
\usepackage{amsthm}
\usepackage{color}

\usepackage{dcolumn}
\newtheorem{theorem}{Theorem}
\usepackage{txfonts}
\usepackage[caption=false]{subfig}
\usepackage[T1]{fontenc}
\usepackage{multirow}

\begin{document}

\title{Quantum Supremacy through Fock State $q$-boson Sampling with Transmon Qubits}

\author{Chon-Fai Kam}
\email{dubussygauss@gmail.com}
\affiliation{Department of Physics, School of Physics, Mathematics and Computing, The University of Western Australia, Perth 6009, Australia}

\author{En-Jui Kuo}
\email{kuoenjui@nycu.edu.tw}
\affiliation{Department of Electrophysics, National Yang Ming Chiao Tung University, Hsinchu, Taiwan, R.O.C.}

\begin{abstract}
Transmon qubits have traditionally been regarded as limited to random circuit sampling, incapable of performing Fock state boson sampling—a problem known to be classically intractable. This work challenges that assumption by introducing $q$-boson Fock state sampling, a variant in which transmon qubits can operate effectively. Through direct mapping to the $q$-boson formalism, we demonstrate that transmons possess the capability to achieve quantum supremacy in $q$-boson sampling tasks. This finding expands the potential applications of transmon-based quantum processors and paves the way for new avenues in quantum computation.
\end{abstract}

\maketitle

\section{Introduction}
Quantum computation aims to design algorithms that surpass classical methods, with Shor’s algorithm for integer factorization offering a theoretical exponential speedup \cite{ekert1996quantum}. However, experimental demonstrations are limited to factoring small numbers like 15 and 21 \cite{lanyon2007experimental, martinis2012shor} due to complex quantum circuits and hardware constraints, requiring approximately $217n^3 \log^2(n)$ CNOT gates  \cite{gidney2021factor} for an $n$-bit integer—about one million for a 10-bit number. These experiments factor numbers classical computers handle instantly, leaving practical quantum advantage unachieved. In contrast, boson sampling, proposed by Aaronson and Arkhipov in 2012 \cite{aaronson2011computational}, uses photonic systems to sample bosonic state distributions, where output probabilities relate to the \#P-hard matrix permanent \cite{valiant1979completeness}. Efficient classical simulation of boson sampling would collapsing the polynomial hierarchy to the third level \cite{stockmeyer1983complexity} even approximately sampling in various statistical distance \cite{bouland2024average, bouland2023complexity}, making it a promising approach for demonstrating quantum supremacy.

Superconducting qubits, such as those used in Google's 53-qubit superconducting processor, Sycamore \cite{arute2019quantum}, which utilizes transmon qubits \cite{koch2007charge}, are well-suited for executing random circuit sampling (RCS), a benchmark task designed to highlight quantum advantage over classical computation \cite{bouland2019complexity}. Unlike boson sampling, which requires precise single-photon inputs and is fundamentally tied to \#P-hard problems like matrix permanents, RCS leverages the complexity of quantum circuits to demonstrate computational supremacy. Superconducting qubits are adept at generating and manipulating deep random quantum circuits but are not configured for Fock-state sampling, which is a requirement in boson sampling experiments \cite{brod2019photonic}. Their architecture, based on microwave-driven superconducting loops, makes them ideal for gate-based quantum operations rather than direct photon-number-state manipulations \cite{kjaergaard2020superconducting}.
Despite these constraints, superconducting qubits remain one of the leading candidates for scalable quantum computing \cite{kjaergaard2020superconducting}. Transmon qubits, in particular, improve coherence times by reducing sensitivity to charge noise, utilizing Josephson junctions within superconducting circuits to perform high-fidelity quantum gate operations \cite{kjaergaard2020superconducting, krantz2019quantum}. The Sycamore experiment, which achieved quantum supremacy in 2019, demonstrated that a 53-qubit transmon-based superconducting processor could complete a random circuit sampling task in just 200 seconds—an operation initially estimated to take Summit, a classical supercomputer, around 10,000 years \cite{arute2019quantum}. However, ongoing advances in classical computing techniques have significantly reduced the time required for simulating these quantum circuits \cite{pednault2019leveraging, huang2020classical}, fueling debates on the extent of quantum supremacy \cite{preskill2012quantum}.

A transmon is a type of superconducting qubit used in quantum computing, particularly in circuit quantum electrodynamics (QED). It is a nonlinear quantum oscillator formed by a Josephson junction shunted by a large capacitance, which reduces its sensitivity to charge noise compared to earlier superconducting qubits, such as the Cooper pair box \cite{koch2007charge}. The transmon is designed to operate as a qubit, but its energy spectrum includes multiple levels due to its nonlinear potential \cite{schreier2008suppressing}.

To account for non-dispersive effects, such as stray population leakage to higher energy states during resonator ring-up, a seven-level model is developed \cite{khezri2016measuring}. The truncation to a finite number of levels is necessary to ensure computational tractability and to account for the negligible population of higher energy states in most scenarios. However, for theoretical rigor or emerging applications—such as quantum simulation with transmons in exotic regimes—an infinite-level model may be necessary \cite{mansikkamaki2022beyond}. Here, exotic regimes refer to operating conditions where transmons deviate from standard qubit operation, such as functioning beyond a two-level system, weak coupling, or the dispersive regime, and exhibit behaviors requiring the full complexity of their spectrum \cite{didier2018analytical, ansari2019superconducting}.

In contrast to the prevailing view that transmon qubits are limited to random circuit sampling and cannot perform Fock state boson sampling—a task theoretically proven to be classically hard—this work challenges that perspective. While transmons cannot be used for standard Fock state boson sampling, they are capable of performing $q$-boson Fock state sampling. By directly mapping to the $q$-boson formalism, we demonstrate that transmons can achieve quantum supremacy in $q$-boson sampling tasks. The $q$-boson framework generalizes standard bosonic algebra by introducing a deformation parameter $q$, which modifies the commutation relations. By utilizing the full nonlinear spectrum of the transmon, rather than limiting it to a two-level or seven-level approximation, we can harness these $q$-deformed properties to enable more efficient sampling implementations.

More precisely, studies have established that quantum sampling tasks involving both standard bosons and $q$-deformed bosons are computationally hard for classical computers. This hardness arises because the altered commutation relations introduce unique quantum interference effects that are infeasible to replicate using conventional classical simulations. In contrast, transmon qubits—or more generally, transmon qudits—operating in an infinite-level configuration can naturally map to $q$-bosons, allowing quantum processors to perform sampling tasks that significantly outpace classical systems. As quantum technologies progress, integrating $q$-deformed structures into superconducting circuits may open new possibilities for demonstrating practical quantum advantage in areas such as boson sampling, quantum machine learning, and beyond. Fundamentally, the weak nonlinearity inherent to transmons allows them to be viewed as $q$-bosons, thereby making them viable candidates for Fock state $q$-boson sampling.

\section{Transmon Qubits as nonlinear Kerr oscillator}

Transmons are specially designed superconducting electrical circuits derived from the Cooper Pair Box (CPB), widely used as qubits in quantum computing due to their reduced sensitivity to charge noise. The transmon consists of a superconducting island coupled to a Josephson junction and a gate capacitance, forming a nonlinear LC circuit. Its quantum mechanical behavior is described by the Hamiltonian
\begin{equation}\label{Josephson}
H = \frac{(Q_J - C_g V_g)^2}{2 C_\Sigma} - E_J \cos\phi,
\end{equation}
where $C_\Sigma = C_g + C_J$ is the total capacitance of the island, comprising the gate capacitance $C_g$ and the junction capacitance $C_J$. The term $Q_J$ represents the charge on the island, $V_g$ is the gate voltage applied through the gate capacitance, and $\phi$ is the superconducting phase difference across the Josephson junction. The first term in Eq.\:\eqref{Josephson} corresponds to the capacitive or charging energy, which arises from the electrostatic energy stored in the capacitors, and the second term in Eq.\:\eqref{Josephson} represents the Josephson inductive energy, where $E_J$ is the Josephson energy, characterizing the strength of the tunneling of Cooper pairs across the junction.

To express the Hamiltonian in a more convenient form, one can introduce the charging energy $E_c$ and the number $n$ of the Cooper pairs. The Hamiltonian in Eq.\:\eqref{Josephson} can be rewritten as
\begin{equation}
H = 4 E_C (n - n_g)^2 - E_J \cos\phi,
\end{equation}
where $E_C = e^2 / (2 C_\Sigma)$ represents the charging energy, and $n = Q_J / (2e)$ represents the number of Cooper pairs on the island, where $2e$ is the charge of a Cooper pair. The effective offset charge $n_g = C_g V_g / (2e)$ accounts for the charge induced by the gate voltage, measured in units of Cooper pairs. The phase $\phi$ and the number operator $n$ are conjugate variables, satisfying the commutation relation $[\phi, n] = i$.

The transmon is a specific regime of the CPB, characterized by a large Josephson energy compared to the charging energy ($E_J \gg E_C$), typically achieved by shunting the Josephson junction with a large capacitance. This design reduces the sensitivity of the qubit's energy levels to charge noise, i.e., fluctuations in $n_g$, making transmons more robust than the original CPB. In the transmon regime, the energy spectrum becomes nearly harmonic, resembling that of a weakly anharmonic oscillator, which allows for well-defined qubit states while maintaining sufficient anharmonicity to distinguish the $|0\rangle$ and $|1\rangle$ states from higher energy levels.

To quantize the system, the Hamiltonian is treated in the quantum mechanical framework, where $n$ and $\phi$ are operators. The charging term $4 E_C (n - n_g)^2$ resembles the kinetic energy of a particle in a quadratic potential, while the Josephson term $-E_J \cos\phi$ acts as a periodic potential. In the transmon limit, the energy levels are approximately given by (see Appendix \ref{B} for a derivation)
\begin{equation}\label{Transmon} 
E_m \approx \sqrt{8 E_J E_C} \left( m + \frac{1}{2} \right) - \frac{E_C}{12} \left( m^2 + m + \frac{1}{2} \right),
\end{equation}
where $m\in \mathbb{N}$ labels the energy eigenstates. The anharmonicity, defined as the difference between the $|1\rangle$ to $|2\rangle$ and $|0\rangle$ to $|1\rangle$ transition frequencies, is approximately $-E_C$, which is small in the transmon regime, ensuring that the qubit can be selectively addressed.

In practice, transmons are fabricated using superconducting materials like aluminum or niobium on a substrate, with the Josephson junction typically formed by an Al-AlOx-Al tunnel junction. The large shunting capacitance reduces $E_C$, pushing the system into the transmon regime. This design has made transmons a cornerstone of modern superconducting quantum processors, enabling high-fidelity quantum gates and long coherence times in systems like those developed by IBM, Google, and Rigetti.

In some references, the transmon is described as \cite{khezri2016measuring}
\begin{equation}
        H_q = \sum_{n,k} E_k |n,k\rangle \langle n,k|, \quad E_k = E_0 + \omega_q k - \eta \frac{k(k-1)}{2}
    \end{equation}
    where \(\omega_q = \omega_{10}\) is the qubit frequency, \(\eta = \omega_{10} - \omega_{21}\) is the anharmonicity, and \(k = 0, 1, \ldots, 6\) indexes the energy levels. In contrast, this work adopts a model that includes an infinite ladder of energy levels.
    
\section{The Kerr Oscillator as a $q$-boson oscillator}

The infinite-level Kerr nonlinear oscillator Hamiltonian for a transmon is given by
\begin{equation}
H = \omega a^\dagger a + \frac{K}{2} a^\dagger a^\dagger a a,
\end{equation}
where $\omega$ is the bare frequency of the oscillator, $K$ is the Kerr nonlinearity strength related to the transmon’s anharmonicity, and $a^\dagger$ and $a$ are the standard bosonic creation and annihilation operators satisfying $[a, a^\dagger] = 1$. The energy eigenvalues of this Hamiltonian for a Fock state $|n\rangle$ are thus give by
\begin{equation}\label{boson}
E_n = \omega n + \frac{K}{2} n(n-1).
\end{equation}
The term $\frac{K}{2} n(n-1)$ introduces nonlinearity, causing the energy level spacing to depend on $n$. For example, the transition energy from $|n-1\rangle$ to $|n\rangle$ is: $\Delta E_n\equiv E_n - E_{n-1} = \omega + K(n-1)$. Hence, the Kerr nonlinearity strength can be measured via the difference in energy level spacing: $\Delta E_3-\Delta E_2 = K$. In transmon literature, the Kerr nonlinearity $ K $ is often negative, leading to a reduction in level spacing compared to the standard harmonic oscillator. In a transmon, the anharmonicity is typically small compared to the bare frequency, reflecting its weakly nonlinear oscillator nature. This nonlinear spacing is key to the $q$-boson mapping.

The $q$-boson algebra \cite{arik1976hilbert} is defined by the commutation relations between the $q$-deformed operators $a_q$, $a_q^\dagger$, and the number operator $N$ (see Appendix \ref{A} for an exploration of various types of $q$-boson algebras)
\begin{equation}
a_q a_q^\dagger - q a_q^\dagger a_q = 1, \quad [N, a_q] = -a_q, \quad [N, a_q^\dagger] = a_q^\dagger.
\end{equation}
Here, $q$ is the deformation parameter, with $q \neq 1$ in general. When $q = 1$, the algebra reduces to the standard bosonic case. The $q$-deformed number operator is defined as $[N]_q = \frac{q^N - 1}{q - 1}$, and the action of the $q$-boson operators on Fock states is
\begin{subequations}
\begin{align}
a_q |n\rangle &= \sqrt{[n]_q} |n-1\rangle, \\
a_q^\dagger |n\rangle &= \sqrt{[n+1]_q} |n+1\rangle,
\end{align}
\end{subequations}
where $[n]_q = \frac{q^n - 1}{q - 1}$. For $q \approx 1$, the $q$-deformed number approaches the standard number, $[n]_q \approx n$.

To relate the deformation parameter $q$ to the Kerr nonlinearity strength $K$, one can compare the energy spectrum of the Kerr oscillator to those of the $q$-boson system. Assume the Kerr nonlinearity is weak ($K \ll \omega$), which is characteristic of transmons, as their anharmonicity remains small relative to their bare frequency. We can map the Kerr Hamiltonian to an equivalent $q$-deformed Hamiltonian. A representative $q$-boson Hamiltonian might take the form
\begin{equation}
H_q = \omega [N]_q,
\end{equation}
where $[N]_q$ introduces nonlinearity. The eigenvalues of $H_q$ are 
\begin{equation}
E_n = \omega [n]_q.\label{qboson}
\end{equation}
For small deformations ($q \approx 1$), one may expand $q = 1 + \delta$ with $\delta \ll 1$, obtaining $ [n]_q \approx n + \frac{\delta}{2} n(n-1)$. Comparing Eq.\:\eqref{qboson} to the Kerr oscillator’s spectrum Eq.\:\eqref{boson} yields $\delta \approx +\frac{K}{\omega}$. Thus
\begin{equation}
q \approx 1 + \frac{K}{\omega}.
\end{equation}

In superconducting transmon qubits, the bare frequency $\omega$ (corresponding to the $|0\rangle \to |1\rangle$ transition) is usually in the range of 4--8 GHz (i.e., $\omega/2\pi \approx 4$--$8 \, \text{GHz}$). The anharmonicity $\alpha$ is typically on the order of 100--300 MHz (i.e., $\alpha/2\pi \approx 0.1$--$0.3 \, \text{GHz}$), though it can vary depending on the specific design parameters, such as the Josephson energy $E_J$ and charging energy $E_C$. Since $K \approx -\alpha$, the magnitude of the Kerr nonlinearity is the same as $\alpha$. Thus the ratio between the Kerr nonlinear strength and the bare frequency is 
\begin{equation}
|K|/\omega \approx \frac{2\pi \times (0.1\text{--}0.3) \, \text{GHz}}{2\pi \times (4\text{--}8) \, \text{GHz}},
\end{equation}
where the lower bound is $\frac{0.1}{8} \approx 0.0125$, and the upper bound is $\frac{0.3}{4} \approx 0.075$. Thus, for typical transmons $|K|/\omega \approx 0.01\text{--}0.08$. 

More physically, the anharmonicity $\alpha$ is determined by the ratio $E_J/E_C$, where $E_J$ is the Josephson energy and $E_C$ is the charging energy. For transmons, $E_J/E_C \gg 1$ (typically 20--100), leading to a small relative anharmonicity:
\begin{equation}
\alpha \approx -E_C, \quad |K| \approx E_C,
\end{equation}
and since $\omega \approx \sqrt{8 E_J E_C} - E_C$, the ratio $|K|/\omega$ is small
\begin{equation}
    \frac{K}{\omega} \approx \frac{1}{\sqrt{8E_J/E_c}-1}
\end{equation}

For example, a transmon with $\omega/2\pi = 6 \, \text{GHz}$ and $\alpha/2\pi = -200 \, \text{MHz}$ has:
\begin{equation}
|K| \approx 2\pi \times 0.2 \, \text{GHz}, \quad \omega \approx 2\pi \times 6 \, \text{GHz},
\end{equation}
\begin{equation}
|K|/\omega \approx \frac{0.2}{6} \approx 0.033.
\end{equation}
This value is consistent with experimental reports for transmons used in quantum computing platforms, such as those in IBM or Google’s superconducting qubit architectures.

The ratio range $|K|/\omega \approx 0.01$--$0.08$ aligns with typical transmon parameters reported in the literature, such as in \cite{khezri2016measuring} and \cite{ansari2019superconducting}, where anharmonicities are on the order of a few hundred MHz and bare frequencies are in the microwave regime. For instance: In \cite{khezri2016measuring}, transmon anharmonicities are discussed in the context of multi-level dynamics, with typical values supporting this range. Experimental papers on transmons \cite{burban2010arik} report $E_C/h \sim 100$--$300 \, \text{MHz}$ and $\omega/2\pi \sim 4$--$8 \, \text{GHz}$, yielding similar ratios.

In Fig.~\ref{fig:spectrum_comparison}, we present a comparison between the Kerr oscillator spectrum and the $q$-boson spectrum across varying values of $K/\omega$. The energy spectrum for the Kerr oscillator, Eq.\:\eqref{boson} can be revised as 
\begin{align}
E_n=\omega\left[n+\frac{1}{2}(q-1)n(n-1)\right],\:q=1+\frac{K}{\omega}.
\end{align}
We compare it with the $q$-deformed spectrum $E_n=\omega [n]_q$. Evidently, the two spectra coincide in the limit $q \to 1$, or equivalently when $K/\omega \rightarrow 0$. Within the approximation $|K/\omega|\lessapprox 0.1$, the two spectra are observed to be nearly identical. We also observe a reduction in level spacing for negative Kerr nonlinearity, i.e, $K<0$ or $q<1$, and conversely, an expansion in level spacing for positive Kerr nonlinearity, i.e., $K>0$ or $q>1$.

\begin{figure}[htbp]
    \centering
    \includegraphics[width=0.5\textwidth]{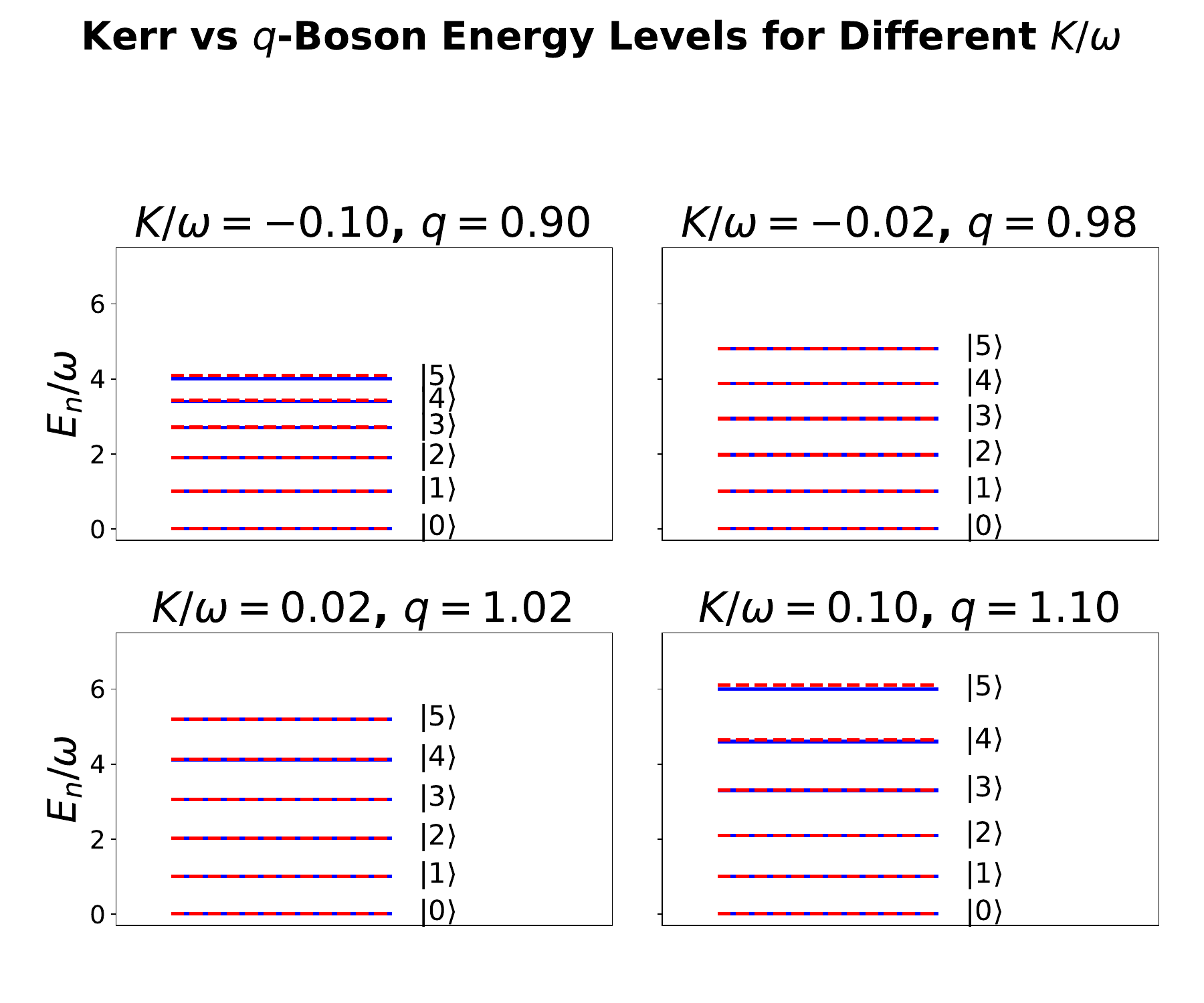}
    \caption{
    Comparison of the renormalized energy level spectra between the Kerr model (solid red lines) and the $q$-boson model (dashed blue lines), evaluated across varying ratios of $K/\omega$. The renormalized energy spectra of the Kerr model and the $q$-boson model coincide at $K/\omega = 0$, and gradually diverge as the nonlinearity parameter $K$ increases. The $q$-boson model captures non-perturbative deviations that extend beyond the Kerr approximation.
    }\label{fig:spectrum_comparison}
\end{figure}

\section{Fock-state $q$-boson sampling}
Generalized bosons are quantum particles in a many-body system on a lattice that obey standard bosonic commutation relations between different sites, i.e., $[a_i, a_j^\dagger] = \delta_{ij}$, $[a_i, a_j] = 0$, and $[a_i^\dagger, a_j^\dagger] = 0$ for $i \neq j$, but deviate from standard bosonic behavior at the same site ($i = j$) \cite{chung1993generalized,Burban2010,MizrahiLimaDodonov2004,Burban2007,BorzovDamaskinskyYegorov1995,Quesne2000,Engquist2009,Brzezinski1992,arik1976hilbert}. Unlike standard bosons, where $[a_i, a_i^\dagger] = 1$, generalized bosons have same-site commutation relations determined by a specific operator diagonal in the local Fock space number basis. For example, $q$-bosons satisfy the same-site commutation relation $\hat{a}\hat{a}^\dagger - q \hat{a}^\dagger \hat{a} = 1$, where $q$ is a deformation parameter \cite{arik1976hilbert}. 

For generalized bosons, the commutation relations $[a_i, a_i] = 0$ and $[a_i^\dagger, a_i^\dagger] = 0$ remain unaltered, while the commutation relation $[a_i, a_i^\dagger]$ is modified by a characteristic function $F(n_i)$, where $n_i$ is the eigenvalue of the number operator $\hat{n}_i = a_i^\dagger a_i$. Specifically, the commutation relation is given by $[a_i, a_j^\dagger] = \delta_{ij} \sum_{n_i=0}^{\infty} F(n_i) |n_i\rangle \langle n_i|$ \cite{kuo2022boson}. The generalized Fock state $|n_i\rangle$ for mode $i$ with occupation number $n_i$ is defined as $|n_i\rangle = \frac{1}{f(n_i)} (a_i^\dagger)^{n_i} |0\rangle$, where $f(n_i)$ is another characteristic function defining the generalized boson. The functions $F(n_i)$ and $f(n_i)$ are related, with their relationship depending on the specific physical system. For standard bosons, we have $F(n_i) = 1$ and $f(n_i) = \sqrt{n_i!}$. For a spin-$S$ system, where $n_i \leq 2S$, we have $F(n_i) = n_i - 2S$ for $n_i \leq 2S$ and zero otherwise, and $f(n_i) = \sqrt{\frac{n_i! (2S)!}{(2S - n_i)!}}$. Therefore, a multi-mode Fock state for the generalized bosons with $m$ modes can be written as \cite{kuo2022boson}
\begin{equation}
    |n_1, n_2, \dots, n_m\rangle = \left(\prod_{i=1}^{m} \frac{1}{f(n_i)}\right) a_1^{\dagger n_1}a_2^{\dagger n_2}\cdots a_m^{\dagger n_m} |0\rangle.
\end{equation}

The Aaronson-Arkhipov Fock-state boson sampling protocol involves sending a Fock state containing $n$ photons across $m$ mode into a randomly selected linear mode-mixing circuit, represented by a unitary matrix $U$ \cite{kuo2022boson}. This process enables mode-mixing, wherein the linear optical network \( U \) transforms the input mode operators \( \{a_i\}_{i=1}^{m} \) into corresponding output operators $b_i = \sum_{j=1}^{m} U_{ij} a_j$. The probability of observing the outcome $\mathbf{k} = (k_1, \dots, k_m)$ when the input configuration is $\mathbf{l}= (l_1, \dots, l_m)$, where $\sum_i k_i = \sum_i l_i = n$, is determined by \cite{aaronson2011computational,kuo2022boson}
\begin{equation}
    \mbox{Pr}(\mathbf{k}|\mathbf{l}) = \frac{\left| \text{Perm}(\Lambda[\mathbf{k}|\mathbf{l}]) \right|^2}{(\prod_i l_i!)( \prod_i k_i!)}.
\end{equation}
Here, the permanent of an $N \times N$ matrix \( A = (A_{i,j})_{i,j} \) defined similarly to the matrix determinant, but differs in that all signs in the summation remain positive
\begin{equation}
\text{Perm}(A) 
\equiv \sum_{\sigma \in S_N} \prod_{i=1}^{N} A_{i, \sigma(i)},
\end{equation}
where $S_N$ denotes the symmetric group on $N$ elements. Unlike the determinant, the permanent lacks alternating signs, making it computationally hard to evaluate. The matrix $\Lambda[\mathbf{k}|\textbf{l}]$ is an $N \times N$ matrix constructed by repeating the \( i \)-th column of \( U \) \( l_i \) times, and the \( j \)-th row \( k_i \) times. 

This protocol can be extended to generalized bosons, enabling applications such as spin-1/2 sampling~\cite{peropadre2017equivalence} and higher-spin sampling~\cite{kam2025beyond}, as demonstrated by Kuo et al.~\cite{kuo2022boson}. Such an extension broadens the scope of Boson sampling to more complex quantum systems. By showing that the Kerr oscillator behaves as a $q$-boson oscillator in the weak nonlinearity regime, defined by the deformed commutation relation $\hat{a}\hat{a}^\dagger-q \hat{a}^\dagger \hat{a} = 1$, the Fock-state sampling protocol of Aaronson and Arkhipov can be generalized to $q$-bosons. This approach can be experimentally realized using infinite-level Kerr nonlinear oscillators, with transmon qubits offering a tunable and practical hardware platform~\cite{kuo2022boson}.

Lastly, in the above, we utilize the fact that off-site modes commute, i.e., \( [a_i, a_j] = 0 \) for all \( i \neq j \). However, in real experimental settings, off-site crosstalk may occur. Such effects have been observed in practical NISQ devices, including those based on IBM architectures~\cite{PhysRevResearch.6.013142, arxiv.2402.06952} and ion trap platforms~\cite{PRXQuantum.2.040338}. Similarly, one could quantify this crosstalk using alternative algebraic frameworks; however, the resulting amplitudes would no longer correspond to permanents, but rather to what one might call $q$ permanent \cite{yang1991q,bapat1994inequalities,lal1998inequalities,da2018mu,da2010mu,de2018noncrossing,tagawa1993q,andelic2018mu}: $\text{Perm}_q(A) 
\equiv \sum_{\sigma \in S_N}q^{\ell(\sigma)} \prod_{i=1}^{N} A_{i, \sigma(i)},$ where $\ell(\sigma)$ is a multivariate combinatorial object that has received considerable attention within the contexts of pure algebra and representation theory. However, its computational complexity remains largely unexplored, and the difficulty of sampling from its associated distribution merits further investigation. Partial results have been obtained \cite{hung2023computational}, but a comprehensive complexity-theoretic characterization is still lacking.

\section{Conclusion and Discussion}
In this work, we have demonstrated that transmon qubits, traditionally viewed as limited to random circuit sampling, can be harnessed for Fock state $q$-boson sampling, a task that is classically intractable and capable of achieving quantum supremacy. By mapping the nonlinear Kerr oscillator Hamiltonian of transmon qubits to the $q$-boson formalism, we have shown that the weak anharmonicity inherent to transmons corresponds to a deformation parameter $q \approx 1 + \frac{K}{\omega}$, where $K$ is the Kerr nonlinearity strength and $\omega$ is the bare frequency. This mapping enables transmon qubits to operate as $q$-boson oscillators, facilitating the implementation of $q$-boson Fock state sampling. Our analysis leverages the full nonlinear spectrum of transmons, moving beyond the conventional two-level or seven-level approximations, to exploit their $q$-deformed properties for quantum computational advantage.

The significance of this result lies in its expansion of the computational scope of transmon-based quantum processors. Unlike standard Fock state boson sampling, which is infeasible with transmon qubits due to their superconducting architecture, $q$-boson sampling capitalizes on the intrinsic nonlinearity of transmons, making them viable candidates for sampling tasks that are computationally hard for classical systems. The $q$-boson framework introduces unique quantum interference effects through modified commutation relations, which we have shown to be naturally suited to the transmon’s energy spectrum. This finding aligns with theoretical results indicating that $q$-deformed systems exhibit computational hardness, reinforcing the potential of transmon qubits to outperform classical computers in specific tasks.

From a practical perspective, transmon qubits offer a robust and tunable platform for implementing $q$-boson sampling. Their well-established use in superconducting quantum processors, as demonstrated in experiments like Google’s Sycamore, ensures compatibility with existing quantum computing infrastructure. The typical anharmonicity of transmons ($|K|/\omega \approx 0.01$--$0.08$) aligns with the weak nonlinearity regime required for $q$-boson mapping, making them readily adaptable for experimental realizations. By extending the Aaronson-Arkhipov boson sampling protocol to $q$-bosons, our work provides a pathway for transmon-based systems to perform sampling tasks that could surpass classical computational capabilities, potentially achieving quantum supremacy in a new domain.

However, several challenges remain. Implementing $q$-boson sampling with transmons requires precise control over the nonlinear spectrum and high-fidelity multi-level operations, which may introduce experimental complexities beyond those encountered in two-level qubit operations. The infinite-level model used in our theoretical framework, while rigorous, assumes idealized conditions that may be affected by noise, decoherence, and higher-order nonlinearities in real-world devices. Additionally, the computational hardness of $q$-boson sampling, while supported by theoretical arguments, requires further validation through experimental demonstrations and comparisons with state-of-the-art classical simulation algorithms. Advances in classical computing, as seen in recent challenges to random circuit sampling claims, underscore the need for robust benchmarks to confirm quantum advantage in $q$-boson sampling tasks.
Future research directions include developing experimental protocols to implement $q$-boson sampling on existing transmon-based platforms, such as those developed by IBM, Google, or Rigetti \cite{kjaergaard2020superconducting, arute2019quantum, rigetti2017cloud}. Optimizing the control of higher energy levels and mitigating leakage to unwanted states will be critical for practical implementations \cite{krantz2019quantum, tripathi2022suppressing}. Furthermore, exploring other generalized boson frameworks, such as those beyond the Arik–Coon oscillator \cite{arik1976hilbert}, could uncover additional applications for transmon qubits in quantum simulation and machine learning. Theoretical studies should also focus on quantifying the computational complexity of $q$-boson sampling in greater detail, potentially by analyzing the scaling of matrix permanents in $q$-deformed systems \cite{brzezinski1993q} or developing hybrid protocols that combine $q$-boson sampling with other quantum advantage tasks \cite{schuld2015introduction}.

In conclusion, this work establishes transmon qubits as versatile tools for $q$-boson Fock state sampling, challenging the notion that they are limited to random circuit sampling. By leveraging their nonlinear properties within the $q$-boson framework, we open new avenues for demonstrating quantum supremacy and advancing quantum computation. As quantum hardware continues to mature, transmon-based $q$-boson sampling could play a pivotal role in realizing practical quantum advantages, bridging the gap between theoretical quantum complexity and experimental feasibility.


\begin{appendix}
\section{Different types of $q$-Bosons}\label{A}

In standard quantum mechanics, the canonical commutation relations, form the foundation of quantization. These relations give rise to the algebra of creation and annihilation operators, which define the behavior of the harmonic oscillator. However, for those physical systems exhibiting nonlinearity, interactions at the Planck scale, or underlying discrete or non-commutative geometries—cannot be adequately captured by this linear structure. As a result, generalizations of the Heisenberg algebra have been proposed, introducing deformed oscillator algebras. These often incorporate a deformation parameter (e.g., $q$ in $q$-deformations), altering the standard bosonic creation and annihilation relations while maintaining certain symmetry or representation-theoretic features.

The Arik–Coon oscillator algebra is one of the earliest and most studied examples of a $q$-deformed oscillator algebra \cite{arik1975operator,arik1976hilbert,burban2010arik}. In the Arik–Coon oscillator algebra, the operator $\hat{a}^\dagger \hat{a}$ plays a role analogous to that of the number operator in the standard harmonic oscillator, but with a deformation that modifies its spectrum. Specifically, the algebra is defined by the relation
\begin{equation}
\hat{a}\hat{a}^{\dagger}-q\hat{a}^{\dagger}\hat{a} = 1, 
\end{equation}
along with the commutation relations
\begin{equation}
[\hat{N},\hat{a}]=-\hat{a}, [\hat{N},\hat{a}^{\dagger}]=\hat{a}^{\dagger}.
\end{equation}
Here, $\hat{a}$ and $\hat{a}^\dagger$ are respectively the deformed annihilation and creation operators, $\hat{N}$ is the number operator, and $q$ is the deformation parameter. In this framework, one identifies the operator $\hat{a}^\dagger \hat{a}$
 with the $q$-number associated with $N$. That is, one defines
\begin{equation}
\hat{a}^\dagger \hat{a} = [\hat{N}]_q = \frac{1 - q^{\hat{N}}}{1 - q}.
\end{equation}
Here, the generators $\{I,a,a^\dagger,\hat{N}\}$ from a $q$-deformed Arik-Coon oscillator algebra, where the $q$-deformed counterpart of $\hat{a}^\dagger \hat{a}$ is interpreted as a key element in the spectrum generating algebra, as it directly labels the eigenstates in the Fock space and determines the physical spectrum of the oscillator. Even though the numerical values differ from the standard case due to the deformation, the underlying mechanism remains the same—the algebra, through its ladder structure, is responsible for generating the spectrum of eigenvalues. The $q$-deformed Heisenberg algebra is not a Lie algebra in the traditional sense. In a conventional Lie algebra—like the standard Heisenberg algebra—the commutation relations are linear in the generators and satisfy the Jacobi identity with constant structure coefficients. However, when we introduce the $q$-deformation, the commutation relations become nonlinear. In the Arik–Coon formulation, the deformation parameter 
$q$ alters the structure so that the bracket is no longer bilinear in the usual sense. This modified structure does not satisfy the standard axioms of a Lie algebra. Instead, $q$-deformed algebras are typically understood within the broader framework of quantum algebras or quantum groups \cite{arik1976hilbert}, often endowed with additional structures like co-products and antipodes, making them Hopf algebras \cite{kassel1995quantum}. These structures allow them to serve as deformations of the classical Lie algebras while accommodating phenomena like nonlinearity and discrete spectra that are not captured by traditional Lie algebra theory.

When acting on a Fock space eigenstate $|n\rangle$, which is an eigenstate of $\hat{N}$ with eigenvalue $n$, we obtain
\begin{equation}
    \hat{a}^\dagger \hat{a} \, |n\rangle = [n]_q \, |n\rangle = \frac{1 - q^n}{1 - q} \, |n\rangle.
\end{equation}
Thus, in the deformed algebra, the eigenvalues of \(\hat{a}^\dagger \hat{a}\) are not the ordinary nonnegative integers \(n\) but their \(q\)-deformed counterparts \(\,[n]_q\). In the limit \(q\to1\), the \(q\)-number \([n]_q\) reduces to \(n\), thereby recovering the conventional oscillator spectrum. This modification encapsulates the nonlinearity inherent in the deformed system and allows the Arik–Coon oscillator to serve as a prototype for nonlinear quantum systems, much as the standard harmonic oscillator does in the linear regime.

In contrast to the Arik–Coon definition, Biedenharn and Macfarlane independently introduced a modified $q$-deformed oscillator algebra as
\cite{Macfarlane1989,biedenharn1989quantum}
\begin{equation}
\hat{a}\hat{a}^{\dagger}-q\hat{a}^{\dagger}\hat{a} = q^{-\hat{N}},
\end{equation}
where the commutation relations are the same as those for Arik–Coon oscillators, but the product $\hat{a}^\dagger \hat{a}$ is replaced by 
\begin{equation}
\hat{a}^\dagger \hat{a}= [\hat{N}]_q = \frac{q^{\hat{N}} - q^{-\hat{N}}}{q - q^{-1}}.
\end{equation}
When acting on a Fock space eigenstate $|n\rangle$, which is an
eigenstate of $\hat{N}$ with eigenvalue $n$, we obtain
\begin{equation}
\hat{a}^\dagger \hat{a}|n\rangle= [n]_q |n\rangle= \frac{q^n - q^{-n}}{q - q^{-1}}|n\rangle.
\end{equation}
This form has the appealing property that it is symmetric under the transformation $q\leftrightarrow q^{-1}$. Like the Arik–Coon algebra, it also tends to the conventional oscillator algebra in the limit where $q\rightarrow 1$. In some mathematical treatments, the deformation can be seen as occurring in steps. Starting from the standard bosonic commutation relation, one may first introduce a minimal deformation that leads to the Arik–Coon algebra. A further or recursive deformation then yields the Biedenharn–Macfarlane algebra and its multiparameter generalizations. This view highlights that the two formalisms are related, with the Biedenharn–Macfarlane approach providing a more refined structure that incorporates additional symmetry. Both approaches reduce to the ordinary harmonic oscillator when $q\rightarrow 1$, but away from this limit, the different $q$-number definitions lead to different spectra and algebraic properties. These differences have practical implications; for example, in constructing coherent states or analyzing non-classical properties in quantum optics, one formulation or the other may be more convenient.

It is also straightforward to quantify the difference between the standard integer $n$ and its $q$-deformed counterpart $[n]_q = \frac{1 - q^n}{1 - q}$. We define the absolute and relative errors, respectively, as follows
\begin{subequations}
\begin{align}
    \Delta_q(n) &\equiv n - [n]_q = \frac{(1 - q)n - (1 - q^n)}{1 - q}, \\
    \delta_q(n) &\equiv \frac{n - [n]_q}{n} = \frac{\Delta_q(n)}{n}.
\end{align}
\end{subequations}
For values of $q$ in proximity to unity, specifically when $q = 1 - \epsilon$ with $\epsilon \ll 1$, one may invoke a Taylor expansion to derive the following approximations
\begin{subequations}
    \begin{align}
    [n]_q &\approx n - \frac{n^2}{2}(1 - q), \\
    \Delta_q(n) &\approx \frac{n^2}{2}(1 - q), \\
    \delta_q(n) &\approx \frac{n}{2}(1 - q).
\end{align}
\end{subequations}
These estimates clearly show that the absolute error $\Delta_q(n)$ exhibits quadratic growth with respect to $n$, whereas the relative error $\delta_q(n)$ grows linearly. As anticipated, both expressions tend to zero in the classical limit $q \to 1$.
Furthermore, given that $[n]_q$ is monotonically increasing and concave in $n$ for $0 < q < 1$, it follows that the deviation from linearity can be tightly bounded above
\begin{equation}
    0 < \Delta_q(n) < \frac{n(n-1)}{2}(1 - q).
\end{equation}
This, in turn, yields a corresponding upper bound on the relative error.”
\begin{equation}
    0 < \delta_q(n) < \frac{n - 1}{2}(1 - q).
\end{equation}
These bounds substantiate the intuitive expectation that the \( q \)-deformed integer \( [n]_q \) becomes a less accurate approximation of the classical integer \( n \) as either \( n \) increases or \( q \) diverges from unity. This trend is illustrated in Fig.~\ref{fig:q-deformed-full}, which presents a comparative analysis of \( n \), \( [n]_q \), and the corresponding relative error \( \delta_q(n) \) across a range of \( q \) values.
\begin{figure}[h]
    \centering
\includegraphics[width=1\linewidth]{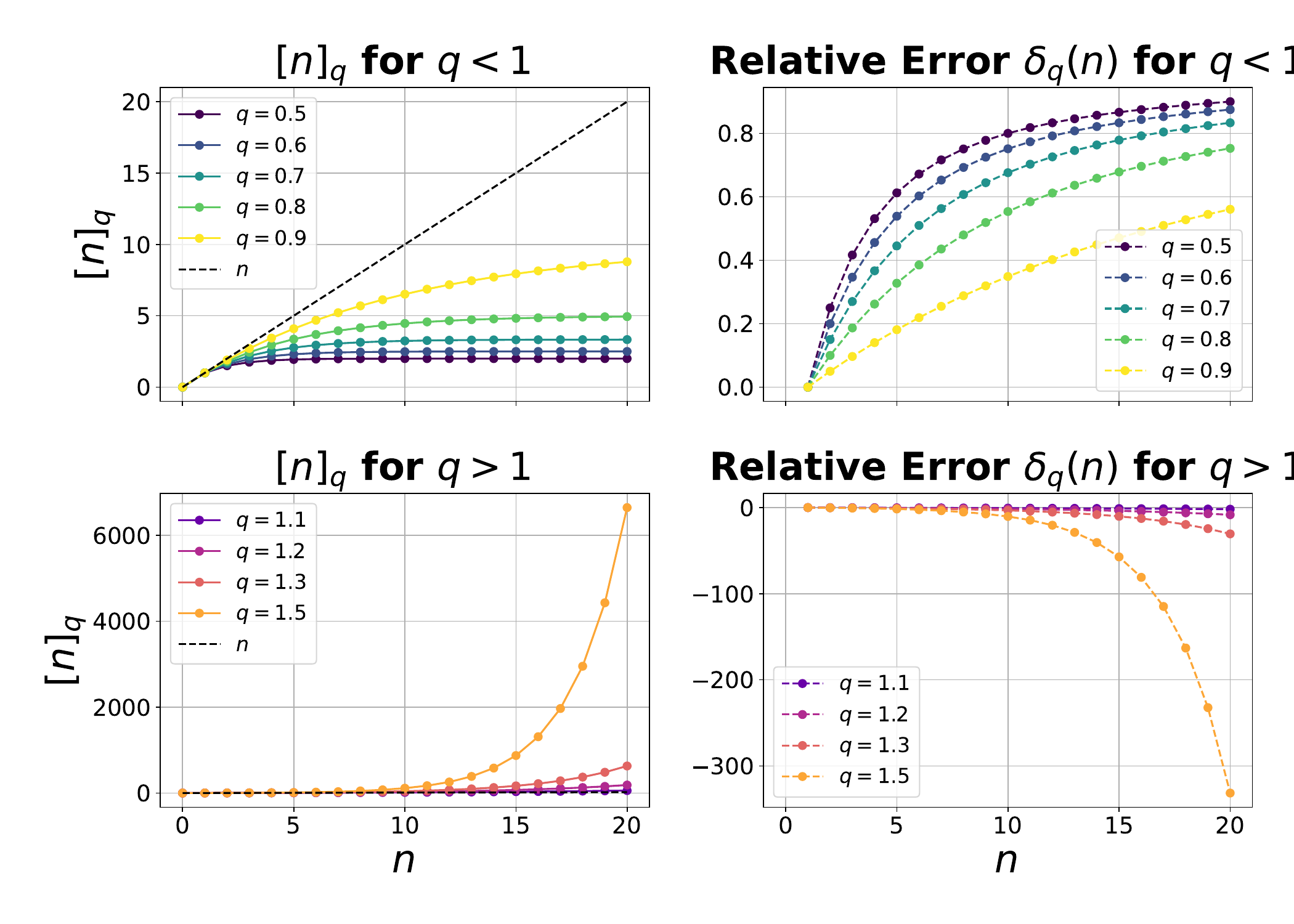}
    \caption{
        This four-panel visualization illustrates the behavior of \( q \)-deformed integers, defined as \( [n]_q = \frac{1 - q^n}{1 - q} \), along with their relative error \( \delta_q(n) = \frac{n - [n]_q}{n} \). The top row corresponds to \( q \in [0.5, 0.9] \), where \( [n]_q \) saturates, and the relative error decreases as \( q \to 1 \). In contrast, the bottom row shows the case for \( q > 1 \), where \( [n]_q \) grows superlinearly and the relative error becomes negative. In all regimes, the classical limit \( q \to 1 \) recovers the standard integer, i.e., \( [n]_q \to n \).
    }
    \label{fig:q-deformed-full}
\end{figure}
We conclude this section by re-examining the distinction between the symmetric \( q \)-boson formulation and the conventional definition. Let \( q = 1 + \delta \), where \( |\delta| \ll 1 \), and define the corresponding difference as follows
\begin{align}
    \Delta_{\mathrm{sym}}(n) \equiv [n]_q^{\mathrm{sym}} - [n]_q,
\end{align}
where $[n]_q^{\mathrm{sym}} = \frac{q^n - q^{-n}}{q - q^{-1}}$ and $[n]_q = \frac{1 - q^n}{1 - q}$. A Taylor expansion in the small parameter \( \delta \), where \( q = 1 + \delta \) with \( |\delta| \ll 1 \), yields
\begin{align}
    [n]_q^{\mathrm{sym}} \approx n + \frac{n^2}{2}\delta, \quad [n]_q \approx n + \frac{n(n-1)}{2}\delta,
\end{align}
so that the difference becomes
\begin{align}
    \Delta_{\mathrm{sym}}(n) \approx \frac{n}{2}\delta + \mathcal{O}(\delta^2).
\end{align}
This indicates that the symmetric \( q \)-number slightly overestimates the standard one for \( \delta > 0 \), even at \( n = 1 \). Consequently, the deviation is generally of order \( \mathcal{O}(\delta) \), and the agreement with \( [n]_q \) at low levels is not exact without additional modifications. Reference~\cite{burban2010arik} presents a unification of various deformed oscillator algebras through the introduction of the structure \( (q; \alpha, \beta, \gamma; \nu) \), defined by 
\begin{equation} \label{burban: struct}
 f(n) =f(0)q^{\gamma n}
+\begin{cases}
q^{\beta}\Bigl(\frac {q^{\gamma n} - q^{\alpha n}} {q^{\gamma} -
q^{\alpha}} + 2\nu \frac {q^{\gamma n} - (-1)^n q^{\alpha n}}
{q^{\gamma} + q^{\alpha}}\Bigr),&\text{if $\alpha\ne \gamma;$}\\
n q^{\gamma(n-1) + \beta} + 2 \nu q^{\gamma(n-1) + \beta}
\Bigl(\frac{1 - (-1)^n}{2}\Bigr),&\text{if $\alpha =
\gamma.$}\end{cases}
\end{equation}
The above defines a generalized bosonic algebra known as the \( (q; \alpha, \beta, \gamma; \nu) \)-deformed oscillator algebra. Below, we provide a general description of how this algebra can be related to the standard \( q \)-boson framework within the context of our sampling protocol.

In particular, we establish the following theorem, which outlines the necessary conditions under which this generalized algebra reproduces the standard \( q \)-integer \( [n]_q \) up to order \( \mathcal{O}(\delta^2) \) for \( n = 0, 1 \). Throughout this analysis, we restrict attention to the case \( \alpha \neq \gamma \).
\begin{theorem}
\label{thm:necessary_conditions}
Let the deformed sequence \( f(n) \) be defined as
\begin{align}
    f(n) := f(0) q^{\gamma n} + q^{\beta} \left( \frac{q^{\gamma n} - q^{\alpha n}}{q^\gamma - q^\alpha} + 2\nu \cdot \frac{q^{\gamma n} - (-1)^n q^{\alpha n}}{q^\gamma + q^\alpha} \right),
\end{align}
where \( \gamma, \alpha \in \mathbb{R} \) with \( \gamma \neq \alpha \), and \( \nu, \beta, f(0) \in \mathbb{R} \), under the expansion \( q = 1 + \delta \) with \( \delta \ll 1 \). Suppose the standard q-number is defined as
\begin{align}
    [n]_q := \frac{1 - q^n}{1 - q}.
\end{align}
If the gap defined by
\begin{align}
    \text{gap}(n) := f(n) - [n]_q
\end{align}
satisfies
\begin{align}
    \text{gap}(n) = \mathcal{O}(\delta^2) \quad \text{for } n = 0 \text{ and } n = 1,
\end{align}
then it is necessary that
\begin{align}
    \boxed{f(0) = 0, \quad \beta = 0, \quad \nu = 0, \quad \gamma + \alpha = 1}.
\end{align}
\end{theorem}
\begin{proof}
We expand all relevant terms assuming \( q = 1 + \delta \), with \( \delta \ll 1 \).
\paragraph{Step 1: Expand the standard q-number.}
\begin{align}
    [n]_q = \frac{1 - (1 + \delta)^n}{-\delta} = n + \frac{n(n - 1)}{2} \delta + \mathcal{O}(\delta^2).
\end{align}
\paragraph{Step 2: Expand \( f(n) \).}
We expand the components
\begin{align}
    q^{\mu n} = (1 + \delta)^{\mu n} = 1 + \mu n \delta + \frac{\mu^2 n^2 - \mu n}{2} \delta^2 + \mathcal{O}(\delta^3),
\end{align}
which yields
\begin{subequations}
    \begin{align}
    q^{\gamma n} - q^{\alpha n} &= (\gamma - \alpha) n \delta + \frac{(\gamma^2 - \alpha^2)n^2 - (\gamma - \alpha)n}{2} \delta^2 + \mathcal{O}(\delta^3), \\
    q^{\gamma} - q^{\alpha} &= (\gamma - \alpha) \delta + \frac{(\gamma^2 - \alpha^2) - (\gamma - \alpha)}{2} \delta^2 + \mathcal{O}(\delta^3).
\end{align}
\end{subequations}
The non-smooth oscillating term
\begin{align}
    \frac{q^{\gamma n} - (-1)^n q^{\alpha n}}{q^\gamma + q^\alpha}
\end{align}
contributes a non-vanishing constant at \( n = 1 \):
\begin{align}
    \text{osc}(1) = 2\nu \cdot \frac{q^\gamma + q^\alpha}{q^\gamma + q^\alpha} = 2\nu.
\end{align}
\paragraph{Step 3: Impose conditions on gap.}
At \( n = 0 \), we require:
\begin{align}
    \text{gap}(0) = f(0) + \mathcal{O}(\delta) = 0 \Rightarrow f(0) = 0.
\end{align}
At \( n = 1 \), we substitute:
\[
f(1) = q^\beta \left( \frac{q^\gamma - q^\alpha}{q^\gamma - q^\alpha} + 2\nu \cdot \frac{q^\gamma - (-1) q^\alpha}{q^\gamma + q^\alpha} \right) = q^\beta (1 + 2\nu).
\]
Using \( f(0) = 0 \), and expanding \( q^\beta = 1 + \beta \delta + \mathcal{O}(\delta^2) \), we find:
\[
\text{gap}(1) = \beta \delta + 2\nu + \left( \frac{\gamma + \alpha - 1}{2} \right) \delta + \mathcal{O}(\delta^2).
\]
To ensure \( \text{gap}(1) = \mathcal{O}(\delta^2) \), the \( \delta^0 \) and \( \delta^1 \) coefficients must vanish $\nu = 0$, $\beta = 0$, and $\gamma + \alpha = 1$. \end{proof}
Notice that for $n=0, 1$, the energy gap satisfies \( \leq \mathcal{O}(\delta^2) \), which is compatible with the requirements of our boson sampling protocol~\cite{kuo2022boson,aaronson2011computational}. Accordingly, we demonstrate that the \( (q; \alpha, \beta, \gamma; \nu) \) oscillator algebra is suitable when the conditions \( \alpha + \gamma = 1 \) and \( f(0) = \beta = \nu = 0 \) are imposed, yielding an exact match with the standard distribution to first order.

\section{Derivation of Transmon Energy Levels}\label{B}

To derive the approximate energy levels Eq.\:\eqref{Transmon} for the transmon Hamiltonian $ H = 4 E_C (n - n_g)^2 - E_J \cos \phi $, we analyze the system in the transmon regime, where $ E_J \gg E_C $. The transmon is a weakly anharmonic oscillator, and the energy levels are obtained by treating the Hamiltonian quantum mechanically, approximating the potential, and applying perturbation theory.

To begin, the transmon Hamiltonian is given by
\begin{equation}
H = 4 E_C (n - n_g)^2 - E_J \cos \phi,
\end{equation}
where $n$ is the number operator for Cooper pairs, representing the number of Cooper pairs on the superconducting island. $\phi$ is the phase operator, representing the superconducting phase difference across the Josephson junction. $E_C$ is the charging energy, with $e$ being the electron charge and $ C_\Sigma$ the total capacitance. $E_J$ is the Josephson energy, characterizing the strength of Cooper pair tunneling. $n_g$ is the effective offset charge induced by the gate voltage $V_g$. $ n $ and $\phi$ are conjugate operators, satisfying the commutation relation $[\phi, n] = i$.

In the transmon regime ($ E_J \gg E_C $), the large Josephson energy flattens the charge dispersion, making the energy levels less sensitive to fluctuations in $n_g$. For simplicity, we set $n_g = 0$ (or assume $ n_g $ is tuned to a sweet spot, e.g., $n_g = 0.5$), as the transmon’s energy levels are nearly independent of $n_g$. Thus, the Hamiltonian becomes
\begin{equation}
H = 4 E_C n^2 - E_J \cos \phi.
\end{equation}
The Hamiltonian resembles that of a quantum mechanical particle in a periodic potential. The term $4 E_C n^2$ is analogous to the kinetic energy, where $ n $ is related to the momentum operator, and the term $-E_J \cos \phi$ acts as a periodic potential, with $\phi$ analogous to the position coordinate.

In the transmon regime, $E_J\gg E_C$, so the $\cos \phi$ term dominates, and the phase $\phi$ is confined to small fluctuations around the minimum of the cosine potential (e.g., $\phi \approx 0$). We approximate the potential by expanding the cosine term for small $\phi$
\begin{equation}
- E_J \cos \phi \approx  -E_J + \frac{E_J \phi^2}{2} - \frac{E_J \phi^4}{24} + \cdots.
\end{equation}
Dropping the constant term $ -E_J $ which shifts all energy levels uniformly and can be redefined as the zero of energy, the Hamiltonian becomes
\begin{equation}\label{Phi4}
H \approx 4 E_C n^2 + \frac{E_J \phi^2}{2} - \frac{E_J \phi^4}{24}.
\end{equation}
The first two terms in Eq.\:\eqref{Phi4} resemble the Hamiltonian of a harmonic oscillator, while the last term in Eq.\:\eqref{Phi4} introduces anharmonicity, which is small in the transmon regime.

Let us first consider only the harmonic part
\begin{equation}
H_0 = 4 E_C n^2 + \frac{E_J \phi^2}{2}.
\end{equation}
To express this in a standard form, note that $ n $ and $ \phi $ are conjugate variables. In the phase basis, $ \phi $ is a coordinate, and $ n = -i \frac{d}{d\phi} $ (in units where $ \hbar = 1 $). The Hamiltonian resembles:
\begin{equation}
H_0 = -4 E_C \frac{d^2}{d\phi^2} + \frac{E_J \phi^2}{2}.
\end{equation}
This is the Hamiltonian of a harmonic oscillator with effective mass $m = \frac{1}{8 E_C}$, and spring constant $k = E_J$. The angular frequency of the oscillator is thus
\begin{equation}
\omega = \sqrt{\frac{k}{m}} = \sqrt{8 E_J E_C}.
\end{equation}
The energy levels of a quantum harmonic oscillator are
\begin{equation}
E_m^{(0)} = \omega \left( m + \frac{1}{2} \right) = \sqrt{8 E_J E_C} \left( m + \frac{1}{2} \right),
\end{equation}
where $m\in \mathbf{N}$ is the quantum number. This gives the leading term in the energy expression.

To continue, the $ -\frac{E_J \phi^4}{24} $ term in Eq.\:\eqref{Phi4} introduces anharmonicity and is treated as a perturbation
\begin{equation}
H' = -\frac{E_J \phi^4}{24}.
\end{equation}
In the harmonic oscillator basis, we compute the first-order correction to the energy levels using perturbation theory. The correction is
\begin{equation}
E_m^{(1)} = \langle m | H' | m \rangle = -\frac{E_J}{24} \langle m | \phi^4 | m \rangle,
\end{equation}
where $|m\rangle$ are the harmonic oscillator eigenstates with frequency $\omega = \sqrt{8 E_J E_C}$. To compute $\langle m | \phi^4 | m \rangle$, we express $\phi$ in terms of ladder operators $a$ and $a^\dagger$
\begin{equation}
\phi  = \left(\frac{2 E_C}{E_J}\right)^{1/4} (a + a^\dagger).
\end{equation}
Hence, one needs to compute the expectation $ \langle m | (a + a^\dagger)^4 | m \rangle $. The non-zero terms contributing to $ \langle m | (a + a^\dagger)^4 | m \rangle $ are those terms that contain exactly two annihilation operators and two creation operators. Using the commutation relation $ [a, a^\dagger] = 1 $, the relevant terms are
\begin{subequations}
\begin{align}
\langle m | (a^\dagger a)^2 | m \rangle &= m^2, \\
\langle m | a^\dagger a^2 a^\dagger | m \rangle &= \langle a a^{\dagger 2}a \rangle = m^2 + m, \\
\langle m| a a^\dagger a a^\dagger|m\rangle &= m^2+2m+1,\\
\langle m| a^{\dagger 2}a^2|m\rangle &= m^2-m,\\
\langle m| a^2a^{\dagger 2}|m\rangle &= m^2+3m+2.
\end{align}
\end{subequations}
Summing the above contributions, we find
\begin{equation}
\langle m | (a + a^\dagger)^4 | m \rangle = 6 m^2 + 6 m + 3.
\end{equation}
Thus the perturbation correction is 
\begin{equation}
E_m^{(1)} \approx -\frac{E_C}{2} \left( m^2 + m + \frac{1}{2} \right).
\end{equation}

Combining the harmonic and anharmonic terms, the energy spectrum for transmon qubits is given by \cite{koch2007charge}
\begin{equation}
E_m \approx \sqrt{8 E_J E_C} \left( m + \frac{1}{2} \right) - \frac{E_C}{12} \left( m^2 + m + \frac{1}{2} \right),
\end{equation}
where the leading term represents the harmonic oscillator energies, while the anharmonic correction arises from the \( \phi^4 \) term, reducing the energy and introducing anharmonicity. The resulting anharmonicity is given by
\begin{equation}
E_2 - E_1 - (E_1 - E_0) \approx -E_C.
\end{equation}
This negative shift shows that the energy spacing between successive levels decreases by approximately \( E_C \) as \( m \) increases, a defining feature of the transmon qubit. We notice that higher-order moments of the form \( \langle m| a^{i} a^{\dagger j} |m\rangle \) can be systematically computed using the combinatorial and normal-ordering techniques developed in~\cite{blasiak2007combinatorics,konig2021newton,kim2022normal}.

In this derivation, we assume \( n_g = 0 \) for simplicity. Since the transmon is insensitive to \( n_g \), this is a reasonable approximation. The exact form of the anharmonic term may vary depending on the approximation method, but the given expression is widely accepted in transmon literature. This derivation assumes the transmon regime ($E_J / E_C \sim$ 50-100), where the harmonic approximation and first-order perturbation suffice.

\end{appendix}

\bibliographystyle{unsrt}
\bibliography{sample}

\onecolumngrid

\end{document}